\documentclass{llncs}

\makeatletter%
\@ifclassloaded{elsarticle}%
{
  \newcommand{\makeothertitle}{}
  \def\el{1}
}%
{\newcommand{\makeothertitle}{\maketitle}
  
 \newenvironment{keyword}{}}%
  \makeatother%
  
\usepackage[ngerman,english]{babel}
\usepackage[utf8]{inputenc}
\usepackage{amsmath}

\usepackage[bold,full]{complexity}
\usepackage{graphicx}
\usepackage{verbatim}
\usepackage{xcolor}

\DeclareMathOperator{\dom}{dom}

\newenvironment{proofx}{}

\title{Separation of \P\;and \NP}

\author{Reiner Czerwinski \inst{1}}
\institute{TU Berlin (Alumnus)}

\begin{document}
\makeothertitle
\begin{abstract}
  There have been many attempts to solve the $\P$ versus $\NP$\;problem. However, with a new proof method,% already used in \cite{czerwinski2021separation},
$\P\not=\NP$ can be proved. 

  A time limit is set for an arbitrary Turing machine and
  an input word is rejected on a timeout. The time limit goes toward $\infty$. Due to the halting problem, whether a word is accepted can only be determined at run time.
  It can be shown by Rice's theorem,
    if a finite set of words are to be checked, they all have to be tested by brute force.

\end{abstract}
\ifdefined\el
\begin{keyword}
  \P\; vs. \NP \sep halting problem \sep Rice's theorem
\end{keyword}
\fi
%\makeeltitle

%In diesem Beweis wird eine Universale Turing Maschiene (UTM) verwendet,
%die zusätzlich noch eine zeitliche Beschränkkung als Eingabe bekommt.

\section{Introduction}

There have been many attempts to solve the $\P$ versus $\NP$ problem.
Since the Cook-Levin theorem \cite{cook1971complexity}, \cite{levin1973universal} in the early 70s, the problem has been the most important unsolved problem in computer science. 
Unfortunately, most conventional methods of proof have failed. Relativizing proof \cite{baker1975relativizations}, natural proof \cite{natural} or algebrization\cite{aaronson2008new} could not solve the problem, and it was concluded that there is no simple proof for $\P\neq\NP$. However, this is a fallacy.
% Notably, many restrictions are known to the proof of the $\P$ versus $\NP$ problem, but not that the proof could not be simple.
There are many barriers to proofs of the P versus NP problem. But this does not mean that there is no simple proof.

In \cite[page 24]{rudich2004computational} is mentioned that the set\\
$\{(M,x,1^t)\mid \text{ NTM }M\text{ accepts }x\text{ within }t\text{ steps}\}$
is $\NP$-complete.% problem to run any non-deterministic TM with an unary-coded time limit.
A similar problem is considered in this paper. Does a deterministic TM accept any word with length $n$ within $t$ steps if $n$ and $t$ are unary coded?
\[ \{ (M,1^n,1^t) \mid \text{ TM } M \text{ accepts some } y \in \{0,1\}^n \text{ within } t \text{ steps }\} \]
This set is $\NP$-complete.
We analyze what happens when $t \to \infty$. With Rice's theorem, we show that one
has to test every $y$ in the worst case.
It is undecidable which $y$ are accepted by $M$. So it is also undecidable which $y$ will be accepted first for $t \to \infty$.

Although we use diagonalization indirectly via the halting problem, the proof is not relativizing. The proof is consistent with $\P^\EXP = \NP^\EXP$. It does not violate the relativization barrier.

% In this paper, we show that one can construct a problem in \P{} by checking whether any Turing machine accepts an input word within a given time.
% Thus, the problem in \NP{} is whether any word of a given length is accepted within a given time. Appropriate coding of the set is assumed. However, it can be shown that one has to check every word in the set in the worst case.

%Now empty your mind. Forget, what you know about the \P-\NP Problem and believe it
%as a problem, that can be solved by a bachelor student.
%
%The idea of the proof is simple. It requires a problem in P. It asks if a word is accepted by a Turing machine. This problem is undecidable. To make it decidable, a time limit is set. So that the problem is not only decidable but also in P, the time limit is coded unary and appended to the input word. Thus the run time is polynomial with respect to the input length.

% In order to find a problem in NP that is related to the problem in P, the program asks whether a word from a set of words is accepted by the Turing machine within the time limit. We let the time limit go towards infinity. Based on Rice's theorem we show that for each word at run time we have to calculate if it is accepted by the Turing machine. Thus the run time is exponential in the worst case.
\subsection{Procedure} \label{myproblem}
In this paper, the problem of determining whether a tuple in the set is considered:
\[ S = \{(M,1^n,1^t) \mid \text{ DTM } M \text{ accepts some } y \in \{0,1\}^n \text{ within } t \text{ steps} \}
\]
is contained. The problem is in \NP.

%Section \ref{lintime} shows that for a given $y$ in the worst case, there are
%$t$ steps to decide whether $y$ is accepted by $M$.
In section \ref{lintime} we analyze the time complexity of the set \[\{(M,x,1^t)\mid \text{ TM }M\text{ accepts }x\text{ within }t\text{ steps }\}\] This is equivalent to the problem $(M,1^0,1^t)\in S$ with arbitrary but fixed $M$.

Section \ref{pnp} shows that in the worst case, all $y\in\{0,1\}^n$ must be checked to decide whether $(M,1^n,1^t)\in S$. This means that the running time would be exponential.

\newcommand{\tmax}{T}%{\tilde{t}}
\newcommand {\UTM}{U}
\section{Preliminaries}
Let $M$ be a deterministic Turing machine with input alphabet $\Sigma=\{0,1\}$.
It is defined, that
$L(M) := \{x \mid M \text{ accepts } x \}$.
In this paper, we investigate what happens when the
running time of a TM is restricted. We additionally define:
\[ L(M,t) := \{x \mid M \text{ accepted } x \text{ within } t \text{ steps}\} \]
%Two-tape TMs are used as Turing machines as in \cite{goldreich2008computational} since they are faster than single-tape machines.

Single-tape TMs are used as Turing machines, as described in \cite{cook2006p}
with input alphabet $\Sigma =\{0,1\}$.

\section{Time Complexity of $L(M,t)$}\label{lintime}
In this section, the running time is considered to compute $L(M,t)$.
It is concluded that to check whether a TM $M$ in $t$ steps
accepts a word, $t$ steps must be calculated in the worst case.
Unfortunately, it takes $0(t)$ steps to find out for fixed $M$ and $x$ whether $x\in L(M,t)$ due to the linear speedup theorem. See \cite[p. 32]{papa1994}.

%\begin{comment}
\begin{definition}
  Let $M$ be a TM and $x$ be an input word. Then,
  $t_M(x)$ is the running time of $M$ on input $x$ if $M$ terminates with it.
  $t_M$ is a partial function. $t_M$ is total iff $M$ terminates on every $x$. 
  And
  \[
    \tmax_M(n) := \max\bigl(\{t_M(x) \mid |x| \le n \text{ and } x\in\dom(t_M) \}\cup\{0\}\bigr)
    \]
\end{definition} see \cite[Chapter 2.4]{wegener2005complexity}.
\newcommand{\withtht}[1]{#1}
\newcommand{\withrec}[1]{\textcolor{green}{} }

 \newcommand{\del}[1]{}
 
 \newenvironment{proofremark}{\sffamily}{}
% \newenvironment{proofremark}{\del}{}
%\newenvironment{proofremark}{\begin{comment}}{\end{comment}}

% \textbf{
% \begin{proofremark}
% Das folgene Lemma wird bewiesen mit den \withtht{Zeithierachiesatz} und dann
% mir dem \withrec{Rekursionstheorem}.
% Letzteres bietet zwar eine bessere Schranke, der Beweis sieht aber nicht ganz
% sauber aus. Ich weiss auch nicht, wie sich das mit dem Speedup-Theorem verträgt.
% \end{proofremark}
% \newcommand{thtending[1]{\withtht{log^2({#1})}}
  \begin{lemma} \label{lin}
    % Let $M$ be a \withtht{multi-taped} %single-tape
    % TM with input word $x$, both arbitrary but fixed.
    % If for a variable $t \in \bbbn$ decide
    % whether $x$ of
    % $M$ is accepted within $t$ steps,
    % i.e. whether $x \in L(M,t)$, then one needs
    % a computation time of $\Omega(t\withtht{/\log^2(t)})$ in the worst case. 

    % Let $M$ be a two-tabe Turing Machine and $w$ an input word, both arbitrary but fixed.
    % Let $t$ grow toward $\infty$.
    % %If the Turing Machine $M$ does not accept the word $w$,
    % The lower bound of time for calculating $L_M(w,t)$ is $\Omega(t)$ in the worst case, i.e. $M$ does not accept $w$.
For an arbitrary \withtht{multi-taped} TM $M$ with input $x$, it takes $\Omega(t\withtht{/\log^2(t)})$ computation time in the worst case to find out whether $x$ is accepted by $M$ within $t$ steps, i.e., whether $x\in L(M,t)$.
  \end{lemma}

  \withtht{%\color{blue}
  \begin{proof}
    Let $T$ be a fully time-constructible function. According to the
    time hierarchy theorem\cite{HLS65}, there exists a TM $M$ with $L(M)\in \DTIME(T)$, but $L(M)\not\in \DTIME(T')$ if $T' * \log(T')\in o(T)$\cite[p. 112]{Homer2011}. Let $T'=T/\log^2(T)$. Then 
  $T'*\log(T')=T/\log^2(T) *\log(T/\log^2(T))\le T/\log(T)\in o(T)$.
\\Let $x_n$ with $n\in\bbbn$ be a sequence with $|x_n| \le n$ and $T_M(n) = t_M(x_n)$. To check, if $x_n\in L(M,T_M(n))$ one need $\Omega(T_M/\log^2(T_M))$ computation time.   
\end{proof}
}
\withrec{%
  \begin{proof}
Assume that for each TM $M$ there is a way to decide whether $x \in L(M,t)$ with less than $\Omega(t)$ steps. Then there would be a constructible TM $M'$ and 
$ x \in L(M,t) \Leftrightarrow x\in L(M',t')$ with
$\lim_{t\to\infty}\frac{t'}{t}=0$. Now let
\[\phi_M(x,t)=\begin{cases}
    1 & \text{ if } x\in L(M,t)\\
    0 & \text{ if } x\not\in L(M,t)\\
  \end{cases}
\]. Since $M'$ is constructible, there is a computable function $f$ with $M'=f(M)$
According to the recursion theorem there is a TM $N$ with $\phi_{f(N)}=\phi_{N}$. Thus $ x \in L(N,t) \Leftrightarrow x\in L(f(N),t)$ holds. This is valid if $L(N)=\emptyset$.Unfortunately, it is undecidable whether $x\in L(N)$. For any $x$ it must be explicitly calculated whether $x\in L(N,t)$. $L(N,t)$ would not be checked faster.
%\end{proof}
%}
% \begin{comment}
%   Let $Q$ be a recursive set. Then, there is a TM $M$,
%   which terminates on every input, with $L(M)=Q$.
%   Thus, both $t_M$ and $\tmax_M$ are total.
%   Now let $(x_n)_{n \in\bbbn}$ be a sequence with
%   $|x_n|\le n$ and $t_M(x_n) = \tmax_M(n)$.
%   If for a sequence element it is to be checked whether
%   $x_n\in L(M,\tmax_M(n))$, then $M$ must $\tmax_M(n)$ steps
%   compute.
%   Suppose it could be computed faster, then there would be a
%   TM $M'$ and a sequence $t'_n$ with
%   $\forall n\in\bbbn : x_n\in L(M',t'_n)\Leftrightarrow x_n\in L(M,\tmax_M(n))$
%   but $\lim_ {n\to\infty}\frac{t'_n}{\tmax_M(n)}$.
%   By Theorem \ref{mintime}, the existence of a TM $M'$ is undecidable.
%   Thus, in the worst case $\Omega(\tmax_M(n))$ calculation steps are necessary to
%   check whether $x_n\in L(M,\tmax_M(n))$ is necessary or not,
%   because we do not know, if a faster TM for this problem exists.
% \end{comment}
\end{proof}  
}

\begin{remark}
Lemma \ref{lin} % applies to both single-tape and multi-tape TMs.
can also be applied to single taped TMs as a special case of multi-taped TMs.
\end{remark}

% The calculation whether a word is accepted within $t$ steps by a TM $M$ 
% takes $\Omega(t)$ steps even with a multi-tape TM in the worst case.
% \begin{proposition}
%   Let $M_1$ be a single-tape TM then it is undecidable
%   whether there is a multi-tape TM which computes $L(M_1)$ faster
%   than $M_1$. E. g. it is undecidable if there is a multi-tape TM $M_2$ with
%   $L(M_2)=L(M_1)$ and $\lim_{n\to\infty}\frac{\tmax_{M_2}(n)}{\tmax_{M_1}(n)}=0$.
% \end{proposition}
% \begin{proof}
%   The proof works analogously to Theorem \ref{mintime}.
%   We consider the running time as a partially computable function.

%   On the one hand, a single-tape TM and a multi-tape TM
%   which accepts the language $R=\{0\}^*$;
%   both have at least linear
%   running time. On the other hand, by Theorem \ref{mintime} it is even undecidable
%   whether there are faster single-tape TMs. A single-tape TM can be simulated be a multi-tape TM by just using one tape.
%   Thus, according to Rice's theorem, it is undecidable whether there are faster
%   multi-tape TMs exist.
% \end{proof}

\section{Proof of $\P \neq \NP$}\label{pnp}
\begin{lemma} \label{meng}
 Let $W=\{w_1,...,w_m\}$ be a finite set of words and $M$ be a TM.
 To find out if every $w \in W$ is not accepted by $M$ within $t$ steps,
 i.e. $W \cap L(M,t) = \emptyset$, one need a running time of $\Omega(m*t\withtht{/\log^2(t)})$ in worst case.
 
\end{lemma}
\begin{proof}
  Proof via induction over $|W|$.
  For $W=\{w_1\}$ it has been proven in Lemma \ref{lin}.
  In the inductive step
  let $M$ be an arbitrary deterministic TM and $W =\{w_1,\hdots,w_m\}$ be a finite set
  of input words.

  According to Rice's theorem, it is undecidable whether $W \cap L(M) = \emptyset$ or  $W \cap L(M) = \{w_m\}$ or $W \cap L(M)$ is something else.
  $W$ is finite, so
  \[\exists T_W \forall t \ge T_W \forall w \in W : w \in L(M)
    \Leftrightarrow w \in L(M,t)\]
  Assume that $t\ge T_W$. If one wants to check if $W \cap L(M,t) = \emptyset$, one has to check whether $\{w_1,\dots,w_{m-1}\}\cap L(M,t)=\emptyset  $ and $w_m\not\in L(M,t)$. This requires by inductions hypothesis $\Omega((m-1)*t\withtht{/\log^2(t)})+\Omega(t\withtht{/\log^2(t)})$
  computation time in the worst case.

If $t < T_W$ should be, then this property is undecidable.
So it is valid too that one needs $\Omega(m* t\withtht{/\log^2(t)} )$ computation time.  
% Let $(T_n)_{n\in\bbbn}$ be a sequence with $\lim_{n\to\infty} T_n =\infty$.
% If $x\in L(M)$, then there exists a $T_x$ with $x\in L(M,T_x)$. Thus: 
% \[\forall x\in W\; \exists T_x\; \forall t > T_x : x\in L(M) \Leftrightarrow x\in L(M,t)\]
% Now let $T_W = \max(\{ T_x \mid x \in W \cap L(M)\})$, where $T_W = 0$ when
% $\forall x \in W : x \not\in L(M)$. $T_W < \infty$ because $W$ is finite.

% If $\lim_{n\to\infty} T_n = \infty$, then there is a $T_n$ with
% $T_n > T_W$. $T_W$ could be so large that it is not computable.
% For $x\in W$ and an arbitrary TM $M$, it is undecidable, according to the halting problem, whether
% $x\not\in L(M)$. By Lemma 1, it takes $\Omega(T_n\withtht{/\log^2(T_n)})$ steps to
% test whether $x\not\in L(M,T_n)$ in worst case. Let $X= \{w \in W \mid w \in L(M)\}$.
% As per Rice's theorem, for any $x \in W$, it is undecidable, 
% whether $X=\{x\}$ or any other subset of $W$.
% If $T_n > T_W$ then $\forall x \in W : x\in L(M) \Leftrightarrow x\in L(M,T_n)$.
% Thus, in the worst case, every $x$ must be tested, especially,
% if no word from $W$ is accepted. This results in a runtime of
% $\Omega(|W|*T_n\withtht{/\log^2(T_n)})$.
\end{proof}
\begin{remark}
That each $x\in W$ must be tested independently of the other words from $W$,
can be seen especially if $M$ is a UTM and each $x\in W$ is the code of a TM.
\end{remark}
\begin{proofx}
%  The proof is by induction over $n$.
%  If $W$ has only one element, it is already proved in Lemma \ref{lin}.

%  By induction assumption for $W=\{w_1,...,w_n\}$ in worst case
%  $\Omega(n*t)$ time is needed to determine whether $W \cap L(M,t) = \emptyset$.

%  In the induction step, let $W' = \{w_1,...,w_{n+1}\} = W \cup\{w_{n+1}\} $.
%  If $t\to\infty$, then:
%  \begin{equation} \label{bigt}
%   \exists T\in\bbbn\; \forall t>T\; \forall w\in W' : \bigl(w \in L(M,t) \Leftrightarrow w \in L(M)\bigr).  
% \end{equation}
% For $w \in W'=\{w_1,...,w_{n+1}\}$, it is undecidable whether $w\in L(M)$.
%  It is possible that $W \cap L(M) = \emptyset$ and $w_{n+1} \in L(M)$. Also,
%  $W \cap L(M) = \emptyset$ and $w_{n+1} \not\in L(M)$.
%  Similarly, both 
%  $W \cap L(M) \not= \emptyset$ and $w_{n+1} \in L(M)$ as well as
%  $W \cap L(M) \not= \emptyset$ and $w_{n+1} \not\in L(M)$ are possible.
%  Thus, whether $w_{n+1}$ is in $L(M)$,
%  even if $W \cap L(M)$ should be known is undecidable by Rice's theorem. 

%  For sufficiently large $t$ (see equation(\ref{bigt}))% it must be determined if
%  $w_{n+1}\in L(M,t)$ must be calculated, when
%  $W \cap L(M,t) = \emptyset$.
%  By Lemma \ref{lin}, this calculation requires
%  $\Omega(t)$ time cost. It verifies that $W \cap L(M,t) = \emptyset$
%  requires a time cost of $\Omega(n*t)$, according to the induction assumption.
%  in the worst case. In total, to check whether $W' \cap L(M,t) = \emptyset$,
%  a time cost of $\Omega((n+1)*t)$ in the worst case.

\end{proofx}

Consider the following set from Section \ref{myproblem}:
\[S= \{ (M,1^n,1^t) \mid \text{ DTM } M \text{ and }\exists y \in \{0,1\}^n \text{ with } y\in L(M,t) \}\]
The problem whether $(M,1^n,1^t)$ is in $S$ is $\NP$-complete for an arbitrary $M$.
% As can be easily seen, the problem is whether an arbitrary tuple
% $(M,1^n,1^t)$ contained in $S$ is in \NP. An NTM can be constructed,
% which writes $n$ symbols from $\{0,1\}$ to the input tape of $M$, then $M$
% executes and stops $t$ steps.
Moreover, it holds:
\[
  (M,1^n,1^t) \in S \Longleftrightarrow \exists y \in \{0,1\}^n : y \in \L(M,t)
  \Longleftrightarrow \{0,1\}^n \cap L(M,t) \not= \emptyset
\]
The set $\{0,1\}^n$ has the cardinality $2^n$.
Due to Lemma \ref{meng}, one needs a computation time of $\Omega(2^n*t\withtht{/\log^2(t)})$  to check whether a tuple is in $S$ in the worst case.
%By Lemma \ref{meng}, the computational cost in the worst case is thus:
%\\$\Omega(|\{0,1\}^n|*t\withtht{/\log^2(t)}) = \Omega(2^n*t\withtht{/\log^2(t)})$.

% \begin{tabular}{ccccc}
%   $b$ :&1&1&...&1\\
%        &$\downarrow$&$\downarrow$&...&$\downarrow$\\
%   $w$ :&$w_1$&$w_2$&...&$w_b$\\
% \end{tabular} with $w_i \in \Sigma_M =\{0,1\}$ and $t$ : $111...1$

\section*{Acknowledgments}
I would like to thank Günter Rote for helping me make the paper smoother.
He gave me many helpful tips.
Furthermore, I thank Paul Spirakis, Hermann Vosseler, Gabriel Wolf, and Martin Dietzfelbinger for helping me improve some formulations.
I also thank Richard Lipton, Michael Sipser, Richard Borcherds, Martin Ritzert, Steven Cook, Nicolas Look,Szabolcs Ivan, Valentin Pickel, David Wellner, and Thomas Karbe for their valuable comments, which helped me craft the formulations more vividly.
I would also like to thank everyone who proofread the paper.

The paper was partly translated with www.DeepL.com/Translator (free version)

%Many thanks especially to G\"unter Rote, Hermann Vosseler and Paul Spirakis. With their comments, I was able to eliminate flaws and  simplify the proof. Thanks to Gabriel Wolf, G\"unter Czerwinski and Thomas Karbe for their feedback and to Silke Ceruti for corrections. The text was partly translated by DeepL.Last, and last but not the least, I thank PaperTrue for their help with proofreading.
\nocite{papa1994}
%\bibliography{lit}{}
%\bibliographystyle{plain}
 \bibliographystyle{splncs04}
 \bibliography{lit}
\end{document}